\documentclass[twocolumn,english,10pt,final,conference,letterpaper]{ieeeconf}
\pdfoutput=1
\usepackage{amsmath,amssymb,bbm,stmaryrd,soul,MnSymbol}
\usepackage{tikz}
\usepackage{stfloats}
\usetikzlibrary{shapes,decorations.pathreplacing}
\usepackage{pgfplots}
\usepackage{babel,flushend}
\usepackage[colorlinks,linktocpage=true]{hyperref}
\RequirePackage[hyperpageref]{backref} 

\newcommand{\entropy}[1]{H\!\left(#1\right)}
\renewcommand{\sc}{spatially-coupled}
\newcommand{\rate}{\mathtt{r}}
\newcommand{\x}{\mathsf{x}}

\newcommand{\ha}{\ensuremath{h^{[1]}}}
\newcommand{\hb}{\ensuremath{h^{[2]}}}
\newcommand{\G}{\mathsf{G}}
\renewcommand{\a}{\mathsf{a}}
\renewcommand{\b}{\mathsf{b}}
\newcommand{\e}{\text{e}}
\newcommand{\X}{\mathbf{X}}
\newcommand{\Y}{\mathbf{Y}}

\newcommand{\gexit}{\mathsf{g}}
\renewcommand{\d}{\text{d}}

\newtheorem{lemma}{Lemma}
\makeatletter
\@addtoreset{remark}{theorem}
\makeatother

\makeatletter
\@addtoreset{corollary}{theorem}
\makeatother

\pgfplotsset{tick label style={
font=\Large}}

\begin{document}
\pgfdeclarelayer{background}
\pgfdeclarelayer{foreground}
\pgfsetlayers{background,main,foreground}
\IEEEoverridecommandlockouts
\title{\huge{Universal Codes for the Gaussian MAC \\via Spatial
    Coupling}} \hypersetup{%
  pdfauthor={Arvind Yedla}, 
  pdftitle={Universal Codes for the Gaussian
    MAC via Spatial Coupling}, 
  pdfkeywords={Gaussian MAC, LDPC codes,
    spatial coupling, EXIT functions, density evolution, joint
    decoding, protograph, area theorem}
}

\author{Arvind Yedla, Phong S. Nguyen, Henry D. Pfister, and Krishna R. Narayanan%
\thanks{This material is based upon work supported, in part, by the
  National Science Foundation (NSF) under Grant No. 0747470, by the
  Texas Norman Hackerman Advanced Research Program under Grant
  No. 000512-0168-2007, and by Qatar National Research Foundation.
  Any opinions, findings, conclusions, and recommendations expressed
  in this material are those of the authors and do not necessarily
  reflect the views of these sponsors.}
\\
  Department of Electrical and Computer Engineering, Texas A\&M University\\
  Email: \{yarvind,psn,hpfister,krn\}@tamu.edu}
\maketitle
\begin{abstract}
  We consider transmission of two independent and separately encoded
  sources over a two-user binary-input Gaussian multiple-access
  channel. The channel gains are assumed to be unknown at the
  transmitter and the goal is to design an encoder-decoder pair that
  achieves reliable communication for all channel gains where this is
  theoretically possible. We call such a system \emph{universal} with
  respect to the channel gains.

  Kudekar et al.~recently showed that terminated low-density
  parity-check convolutional codes (a.k.a. \sc{} low-density
  parity-check ensembles) have belief-propagation thresholds that
  approach their maximum a-posteriori thresholds.  This was proven for
  binary erasure channels and shown empirically for binary memoryless
  symmetric channels. It was conjectured that the principle of
  spatial coupling is very general and the phenomenon of threshold
  saturation applies to a very broad class of graphical models.  In
  this work, we derive an area theorem for the joint decoder and
  empirically show that threshold saturation occurs for this
  problem. As a result, we demonstrate near-universal performance for
  this problem using the proposed \sc{} coding system.
\end{abstract}

\thispagestyle{empty}

\pagestyle{empty}

\section{Introduction}
\label{sec:introduction}
The phenomenon of threshold saturation was introduced by Kudekar et
al. \cite{Kudekar-it11} to explain the impressive performance of
convolutional LDPC ensembles
\cite{Felstrom-it99,Lentmaier-isit05}. They observed that the
belief-propagation (BP) threshold of a \sc{} ensemble is very close to
the maximum-a-posteriori (MAP) threshold of its underlying ensemble; a
similar statement was formulated independently, as a conjecture in
\cite{Lentmaier-isit10}. This phenomenon has been termed ``threshold
saturation via spatial coupling''. Kudekar et al. prove in
\cite{Kudekar-it11} that threshold saturation occurs for the binary
erasure channel (BEC) and a particular convolutional LDPC ensemble and
conjecture that the phenomenon of threshold saturation is very
general. Empirical evidence of this phenomenon has also been shown for
additive white Gaussian noise (AWGN) channels in \cite{Kudekar-istc10}
and for a class of channels with memory (the dicode erasure channel)
in \cite{Kudekar-isit11-DEC}, using EXIT-like curves. It is known that
the MAP threshold of regular LDPC codes approaches the Shannon limit
with increasing left degree, while keeping the rate fixed (though such
codes a have a vanishing BP threshold) \cite{Kudekar-it11}. So,
spatial coupling appears to provides us with a new paradigm to
construct capacity approaching codes for BMS channels.

Spatial-coupling is now being applied to more general scenarios. The
noisy Slepian-Wolf problem was considered in \cite{Yedla-isit11} and
the authors showed that the phenomenon of threshold saturation extends
to multi-terminal problems. Threshold saturation was shown for the
binary-adder channel with erasures in \cite{Kudekar-isit11-MAC} by
considering EXIT-like curves. Spatially-coupled codes have also been
shown to achieve the entire rate-equivocation region for the BEC
wiretap channel \cite{Rathi-isit11}. The effect of coupling has also
been observed for $K$-satisfiability, graph coloring and the
Curie-Weiss model of statistical physics in \cite{Hassani-itw10}, and
for compressive sensing in \cite{Kudekar-aller10} by using \sc{}
measurement matrices.

The notion of universality with respect to channel parameters was
considered in \cite[p. 398]{Tse-2005} and \cite{Tavildar-it06} in the
context of compound channels. Yedla et al. focused on the notion of
universality with respect to channel parameters for multi-terminal
problems in \cite{Yedla-aller09} and showed that \sc{} codes are
near-universal for the noisy Slepian-Wolf problem in
\cite{Yedla-isit11}. Preliminary results in \cite{Yedla-isit11} show
that spatial coupling increases the threshold for transmission over
the binary-input Gaussian multiple-access channel (MAC) as well.

In this paper, we consider the performance of \sc{} codes for
transmission over the $2$-user binary-input Gaussian MAC and
investigate the phenomenon of threshold saturation for this
problem. The Gaussian MAC has been extensively studied in the
literature and is defined by
\begin{align}
\label{eq:gmac}
  Y &= \ha X^{[1]} + \hb X^{[2]} + N.
\end{align}
The system model is shown in Fig.~\ref{fig:gmac-sys-model}. The
channel inputs $X^{[1]},X^{[2]}\in\{\pm 1\}$ and the variation in
channel gains $\ha ,\hb \in[0,\infty)$ can be explained either by
fading or by different power constraints for the two users. The noise
$N$ is a zero-mean Gaussian random variable, with fixed variance of
$1$. The capacity region is defined as the set of all
achievable rate tuples $(R^{[1]},R^{[2]})$, given by the equations

\begin{figure}[t!]
  \centering
  \begin{tikzpicture}[>=stealth]
    \node (s1) at (0,2) {$X^{[1]}$};
    \node (s2) at (0,0) {$X^{[2]}$};
    \node[coordinate] (p1) at (1.5,2) {};
    \node[coordinate] (p2) at (1.5,0) {};
    \node[draw,circle,inner sep=2pt] (c)   at (3,1) {$+$};
    \draw[->,thick] (s1) -- node[very near end,above] {$\ha $} (p1) -- (c);
    \draw[->,thick] (s2) -- node[very near end,below] {$\hb $} (p2) -- (c);
    \node (noise) at (3,2.5) {\hspace{4mm}$N\sim \mathcal{N}(0,1)$};
    \draw[->,thick] (noise) -- (c);
    \node (rec) at (5,1) {$Y$};
    \draw[->,thick] (c) -- (rec);
  \end{tikzpicture}
  \caption{The Gaussian MAC}
  \label{fig:gmac-sys-model}
\end{figure}
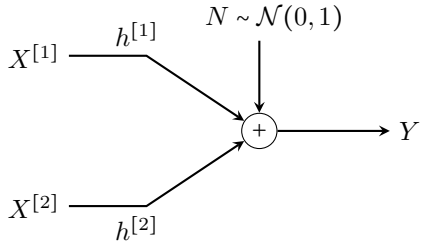

\begin{align}
  \label{eq:gmac_capacity_eq}
  R^{[1]} &\leq I\left(X^{[1]};Y|X^{[2]}\right)\nonumber \\
  R^{[2]} &\leq I\left(X^{[2]};Y|X^{[1]}\right) \\
  R^{[1]} + R^{[2]} &\leq I\left(X^{[1]},X^{[2]};Y\right).\nonumber
\end{align}

The corner points of the capacity region are known to be achievable by
combining successive cancellation at the decoder with single-user
codes \cite{Palanki-aller01}. This method can also be leveraged to
achieve any point on the dominant face by time sharing or rate
splitting \cite{Rimoldi-it96}. The problem of designing good LDPC
degree distributions was studied in \cite{Amraoui-isit02} using
density evolution (DE), where the authors design good LDPC codes for a
few points in the achievable region (in terms of rate). Another
approach was shown in \cite{Roumy-eurasip07} for the case when both
users have the same transmit power, using EXIT charts.

These optimization procedures exploit knowledge of the channel gains
to design good codes. However, in practical scenarios the channel
gains cannot be known non-causally at the transmitter (for example, a
fading channel). So, it is desirable to fix the rate pair for
transmission and view the capacity region in terms of the achievable
channel gains for that rate pair. In other words, the capacity region
is the set of all channel gains $(\ha ,\hb )$ that are achievable,
i.e., satisfy \eqref{eq:gmac_capacity_eq}. We call this region as the
MAC achievable channel-parameter region (MAC-ACPR), to illustrate that
the capacity region is defined in terms of achievable channel
parameters. The achievable channel parameter regions (ACPRs) were
introduced in \cite{Liu-it06} as reliable channel regions, in the
context of communication over parallel channels. As the channel gains
are not known at the transmitter, it is desirable to use codes which
can simultaneously achieve the entire MAC-ACPR (in order to minimize
the outage probability for fading channels).

Coding schemes which can achieve the entire rate region are said to be
\emph{universal}. However, LDPC codes optimized for a fixed $2$-user
binary-input Gaussian MAC need not perform well for different channel
parameters. Fig.~\ref{fig:acpr_mct_opt} shows the belief propagation
achievable channel parameter region (BP-ACPR), computed via density
evolution (DE), of an LDPC ensemble optimized for the equal power case
\cite[p. 311]{RU-2008}. Even though the
performance is close to capacity for the equal power case, the
performance is far from optimal for asymmetric channel gains. So we
need to consider additional constraints when optimizing irregular LDPC
codes to achieve universality. Spatially-coupled codes have been shown
(via numerical computation of the BP-ACPR) to provide near-universal
performance for the noisy Slepian-Wolf problem \cite{Yedla-isit11}. In
this paper, we consider \sc{} codes as a potential candidate to
achieve universality for this problem.

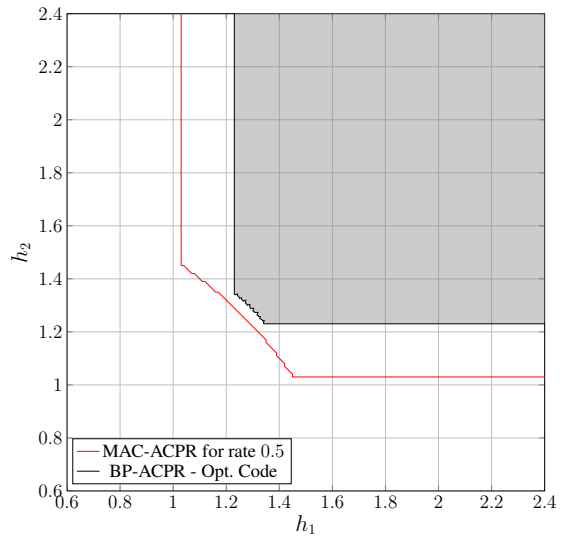
\begin{figure}
  \centering
  \begin{tikzpicture}[scale=0.5]

\begin{axis}[
scale only axis,
width=5in,
height=5in,
xmin=0.6, xmax=2.4,
ymin=0.6, ymax=2.4,
xmajorgrids,
ymajorgrids,
xlabel={\LARGE$h_1$},
ylabel={\LARGE$h_2$},
legend style = {at={(0.47,0.11)}}
]

\addplot [
color=red,
solid
]
coordinates{
 (1.03,1.45)
 (1.04,1.45)
 (1.05,1.44)
 (1.06,1.43)
 (1.07,1.42)
 (1.08,1.42)
 (1.09,1.41)
 (1.1,1.4)
 (1.11,1.39)
 (1.12,1.39)
 (1.13,1.38)
 (1.14,1.37)
 (1.15,1.36)
 (1.16,1.35)
 (1.17,1.35)
 (1.18,1.34)
 (1.19,1.33)
 (1.2,1.32)
 (1.21,1.31)
 (1.22,1.3)
 (1.23,1.29)
 (1.24,1.28)
 (1.25,1.27)
 (1.26,1.26)
 (1.27,1.25)
 (1.28,1.24)
 (1.29,1.23)
 (1.3,1.22)
 (1.31,1.21)
 (1.32,1.2)
 (1.33,1.19)
 (1.34,1.18)
 (1.35,1.17)
 (1.35,1.16)
 (1.36,1.15)
 (1.37,1.14)
 (1.38,1.13)
 (1.39,1.12)
 (1.39,1.11)
 (1.4,1.1)
 (1.41,1.09)
 (1.42,1.08)
 (1.42,1.07)
 (1.43,1.06)
 (1.44,1.05)
 (1.45,1.04)
 (1.45,1.03)
 (1.46,1.03)
 (1.47,1.03)
 (1.48,1.03)
 (1.49,1.03)
 (1.5,1.03)
 (1.51,1.03)
 (1.52,1.03)
 (1.53,1.03)
 (1.54,1.03)
 (1.55,1.03)
 (1.56,1.03)
 (1.57,1.03)
 (1.58,1.03)
 (1.59,1.03)
 (1.6,1.03)
 (1.61,1.03)
 (1.62,1.03)
 (1.63,1.03)
 (1.64,1.03)
 (1.65,1.03)
 (1.66,1.03)
 (1.67,1.03)
 (1.68,1.03)
 (1.69,1.03)
 (1.7,1.03)
 (1.71,1.03)
 (1.72,1.03)
 (1.73,1.03)
 (1.74,1.03)
 (1.75,1.03)
 (1.76,1.03)
 (1.77,1.03)
 (1.78,1.03)
 (1.79,1.03)
 (1.8,1.03)
 (1.81,1.03)
 (1.82,1.03)
 (1.83,1.03)
 (1.84,1.03)
 (1.85,1.03)
 (1.86,1.03)
 (1.87,1.03)
 (1.88,1.03)
 (1.89,1.03)
 (1.9,1.03)
 (1.91,1.03)
 (1.92,1.03)
 (1.93,1.03)
 (1.94,1.03)
 (1.95,1.03)
 (1.96,1.03)
 (1.97,1.03)
 (1.98,1.03)
 (1.99,1.03)
 (2,1.03)
 (2.01,1.03)
 (2.02,1.03)
 (2.03,1.03)
 (2.04,1.03)
 (2.05,1.03)
 (2.06,1.03)
 (2.07,1.03)
 (2.08,1.03)
 (2.09,1.03)
 (2.1,1.03)
 (2.11,1.03)
 (2.12,1.03)
 (2.13,1.03)
 (2.14,1.03)
 (2.15,1.03)
 (2.16,1.03)
 (2.17,1.03)
 (2.18,1.03)
 (2.19,1.03)
 (2.2,1.03)
 (2.21,1.03)
 (2.22,1.03)
 (2.23,1.03)
 (2.24,1.03)
 (2.25,1.03)
 (2.26,1.03)
 (2.27,1.03)
 (2.28,1.03)
 (2.29,1.03)
 (2.3,1.03)
 (2.31,1.03)
 (2.32,1.03)
 (2.33,1.03)
 (2.34,1.03)
 (2.35,1.03)
 (2.36,1.03)
 (2.37,1.03)
 (2.38,1.03)
 (2.39,1.03)
 (2.4,1.03)
 (2.41,1.03)
 (2.42,1.03)
 (2.43,1.03)
 (2.44,1.03)
 (2.45,1.03)
 (2.46,1.03)
 (2.47,1.03)
 (2.48,1.03)
 (2.49,1.03)
 (2.5,1.03)
 (2.5,2.5)
 (1.03,2.5)
 (1.03,2.49)
 (1.03,2.48)
 (1.03,2.47)
 (1.03,2.46)
 (1.03,2.45)
 (1.03,2.44)
 (1.03,2.43)
 (1.03,2.42)
 (1.03,2.41)
 (1.03,2.4)
 (1.03,2.39)
 (1.03,2.38)
 (1.03,2.37)
 (1.03,2.36)
 (1.03,2.35)
 (1.03,2.34)
 (1.03,2.33)
 (1.03,2.32)
 (1.03,2.31)
 (1.03,2.3)
 (1.03,2.29)
 (1.03,2.28)
 (1.03,2.27)
 (1.03,2.26)
 (1.03,2.25)
 (1.03,2.24)
 (1.03,2.23)
 (1.03,2.22)
 (1.03,2.21)
 (1.03,2.2)
 (1.03,2.19)
 (1.03,2.18)
 (1.03,2.17)
 (1.03,2.16)
 (1.03,2.15)
 (1.03,2.14)
 (1.03,2.13)
 (1.03,2.12)
 (1.03,2.11)
 (1.03,2.1)
 (1.03,2.09)
 (1.03,2.08)
 (1.03,2.07)
 (1.03,2.06)
 (1.03,2.05)
 (1.03,2.04)
 (1.03,2.03)
 (1.03,2.02)
 (1.03,2.01)
 (1.03,2)
 (1.03,1.99)
 (1.03,1.98)
 (1.03,1.97)
 (1.03,1.96)
 (1.03,1.95)
 (1.03,1.94)
 (1.03,1.93)
 (1.03,1.92)
 (1.03,1.91)
 (1.03,1.9)
 (1.03,1.89)
 (1.03,1.88)
 (1.03,1.87)
 (1.03,1.86)
 (1.03,1.85)
 (1.03,1.84)
 (1.03,1.83)
 (1.03,1.82)
 (1.03,1.81)
 (1.03,1.8)
 (1.03,1.79)
 (1.03,1.78)
 (1.03,1.77)
 (1.03,1.76)
 (1.03,1.75)
 (1.03,1.74)
 (1.03,1.73)
 (1.03,1.72)
 (1.03,1.71)
 (1.03,1.7)
 (1.03,1.69)
 (1.03,1.68)
 (1.03,1.67)
 (1.03,1.66)
 (1.03,1.65)
 (1.03,1.64)
 (1.03,1.63)
 (1.03,1.62)
 (1.03,1.61)
 (1.03,1.6)
 (1.03,1.59)
 (1.03,1.58)
 (1.03,1.57)
 (1.03,1.56)
 (1.03,1.55)
 (1.03,1.54)
 (1.03,1.53)
 (1.03,1.52)
 (1.03,1.51)
 (1.03,1.5)
 (1.03,1.49)
 (1.03,1.48)
 (1.03,1.47)
 (1.03,1.46)
 (1.03,1.45)

};
\addlegendentry{\Large MAC-ACPR for rate $0.5$}

\addplot [
black,fill=gray,fill opacity=0.4
]
coordinates{
 (1.23047,2.46094)
 (1.23047,1.34277)
 (1.2334,1.3398)
 (1.24219,1.34314)
 (1.24219,1.33693)
 (1.24609,1.3349)
 (1.24805,1.33075)
 (1.24976,1.32633)
 (1.25781,1.32859)
 (1.25781,1.3223)
 (1.26172,1.3201)
 (1.26367,1.31582)
 (1.27344,1.31962)
 (1.27344,1.31326)
 (1.27344,1.30689)
 (1.27734,1.30451)
 (1.2793,1.30011)
 (1.28906,1.30359)
 (1.28906,1.29714)
 (1.28906,1.2907)
 (1.2907,1.28906)
 (1.29714,1.28906)
 (1.30359,1.28906)
 (1.30011,1.2793)
 (1.30451,1.27734)
 (1.30689,1.27344)
 (1.31326,1.27344)
 (1.31962,1.27344)
 (1.31582,1.26367)
 (1.3201,1.26172)
 (1.3223,1.25781)
 (1.32859,1.25781)
 (1.32633,1.24976)
 (1.33075,1.24805)
 (1.3349,1.24609)
 (1.33693,1.24219)
 (1.34314,1.24219)
 (1.3398,1.2334)
 (1.34277,1.23047)
 (2.46094,1.23047)
} |- (axis cs:2.5,2.5) -- cycle;
\addlegendentry{\Large BP-ACPR - Opt. Code}

\end{axis}
\end{tikzpicture}

  \caption{BP-ACPR for an LDPC code optimized for the equal rate,
    equal power case and the MAC-ACPR for rate $1/2$. The degree
    profiles can be found in \cite[p. 311]{RU-2008}.}
  \label{fig:acpr_mct_opt}
\end{figure}

To simplify notation, we assume that both users employ different codes
of rate $R$, from the same ensemble. We use \sc{} codes for this
problem and provide numerical evidence for the phenomenon of threshold
saturation. Analogous to the MAP threshold for point-to-point
communication, one can define a MAP boundary for multi-terminal
scenarios. An appropriate GEXIT kernel is also defined to construct
GEXIT curves for this problem. These GEXIT curves satisfy the area
theorem by definition and provide an outer bound to the MAP
boundary. The threshold saturation phenomenon is observed to occur
towards the MAP boundary i.e., the BP-ACPR boundary saturates towards
the MAP boundary. This provides numerical evidence that supports the
conjecture that the outer bound on the MAP boundary, computed via the
area theorem, is tight. The MAP boundary for the $(4,8)$ regular LDPC
ensemble is observed to be close to the boundary of the MAC-ACPR,
thereby showing that \sc{} codes provide near-universal performance
for this problem. To the knowledge of the authors, this is currently the only
practical coding scheme that achieves near-universal performance for
this problem.

The paper is organized as follows: In
Section~\ref{sec:dens-evol-gexit}, we describe DE for the joint
iterative decoder and the GEXIT curves used to identify the upper
bounds on the MAP threshold are discussed in
Section~\ref{sec:map-threshold}. We briefly discuss spatial coupling
in Section~\ref{sec:spatial-coupling} and density evolution for \sc{}
codes in Section~\ref{sec:density-evolution}. The results are
presented in Section~\ref{sec:results-concl-remark}.


\section{Density Evolution and GEXIT Curves}
\label{sec:dens-evol-gexit}
In this section, we briefly cover some notation and background for
LDPC codes, the $2$-user binary-input Gaussian MAC and density
evolution. We will then discuss GEXIT curves for the joint iterative
decoder and use them to compute an outer bound on the MAP boundary.

\subsection{Background}
\label{sec:background}

An LDPC ensemble can be characterized by its degree profiles. Based on
standard notation \cite{RU-2008}, we let $\lambda(x) = \sum_i
\lambda_i x^{i-1}$ be the degree distribution (from an edge
perspective) corresponding to the variable nodes and $\rho(x) = \sum_i
\rho_i x^{i-1}$ be the degree distribution (from an edge perspective)
of the parity-check nodes in the decoding graph. The coefficient
$\lambda_i$ (resp. $\rho_i$) gives the fraction of edges that connect
to variable nodes (resp. parity-check nodes) of degree
$i$. Likewise, $L_i$ is the fraction of variable nodes with degree
$i$. The design rate of an LDPC ensemble is given by
\begin{align*}
  \rate(\lambda,\rho) = 1 -
  \frac{\int_0^1\rho(x)\text{d}x}{\int_0^1\lambda(x)\text{d}x}.
\end{align*}

\subsection{The joint iterative decoder}
\label{sec:2-user-gaussian}

Let $\X^{[1]},\X^{[2]} \in \{\pm 1\}^n$ denote the input of users $1$
and $2$ respectively and $\Y$ denote the received vector. We consider
the case when both users employ LDPC codes for transmission. To
simplify notation, we assume that both the users encode the data using
LDPC codes from the standard irregular ensemble LDPC$(n,\lambda,\rho)$
and that the transmission is bit-aligned. The factor graph of the
joint decoder (see Fig.~\ref{fig:jd_tanner}) consists of two single
user Tanner graphs, whose variable nodes are connected through a
function node \cite[p. 308]{RU-2008}. The variable nodes that are
connected via the function node are chosen at random\footnote{Other
  matching rules result in a different performance in general.}. We
also assume that the joint decoder iterates by performing one round of
decoding for user one, followed by one round of decoding for user
two. Let $X_i = \left(X^{[1]}_i,X^{[2]}_i\right)$ and $\X =
\left(\X^{[1]},\X^{[2]}\right)$. Without loss of generality, we can
label the elements of $\{\pm 1\}^2$ by integers
$\mathcal{X}\triangleq\{0,1,2,3\}$ using the map
$\pi:\mathcal{X}\to\{\pm 1\}^2$, defined by $0\mapsto (+1,+1), 1 \mapsto (+1,-1),
2\mapsto (-1,+1)$ and $3\mapsto (-1,-1)$. Let
$\pi_1,\pi_2:\mathcal{X}\to\{\pm 1\}$ be the projections onto the
first and second coordinate respectively. Then, the canonical
representation of the channel output is given by
\begin{align*}
  \nu_{x_i}(y_i) &= p_{Y|X^{[1]},X^{[2]}}(y_i|\pi_1(x_i),\pi_2(x_i)) \\
  &= \frac{1}{\sqrt{2\pi\sigma^2}}\exp\left[-\frac{(y_i-\ha \pi_1(x_i)-\hb \pi_2(x_i))^2}{2\sigma^2}\right].
\end{align*}
Let $m^{[j]}_{i,v\to f}$ and $m^{[j]}_{i,f\to v}$ denote the
``variable node to function node'' and ``function node to variable
node'' messages\footnote{Here, the messages are in the log-likelihood
  domain.}, respectively, for variable node $i$ of the $j$th
user. Here $j\in\{1,2\}$ and $i\in\{1,2,\ldots,n\}$. The message
passing rules at the function node are given by
\begin{align}
  \label{eq:fun_12}
  m^{[1]}_{i,f\to v} &= \log\frac{\nu_{0}(y_i)\e^{m^{[2]}_{i,v\to f}} +
    \nu_{1}(y_i)}{\nu_{2}(y_i)\e^{m^{[2]}_{i,v\to f}} + \nu_{3}(y_i)}, \\
  \label{eq:fun_21}
  m^{[2]}_{i,f\to v} &= \log\frac{\nu_{0}(y_i)\e^{m^{[1]}_{i,v\to f}} +
    \nu_{2}(y_i)}{\nu_{1}(y_i)\e^{m^{[1]}_{i,v\to f}} + \nu_{3}(y_i)}.
\end{align}
Note that this function node operation is not symmetric with respect to
the users in general. The operations are the same only for the case of
symmetric fading coefficients i.e., when $\ha  = \hb$. 

\subsection{Density evolution}
\label{sec:density-evolution-1}

Density evolution is employed to analyze the asymptotic performance of
the ensemble LDPC$(n,\lambda,\rho)$. The lack of symmetry means that
one cannot use the all-zero codeword assumption for this
problem. Instead, one may assume that both users transmit codewords of
type one-half, which occurs with high probability (a more thorough
discussion can be found in \cite[p. 296]{RU-2008}). In this case,
$p_X(x) = 1/4,\forall x\in\mathcal{X}$. We use the notation
$\a_{\text{BAWGNMA}}\triangleq\a_{\text{BAWGNMA}(\ha ,\hb )}$ to
denote the density of the received random variable $Y$. Let $f_{1\to
  2}(\cdot,\a_{\text{BAWGNMA}})$ (respectively $f_{2\to
  1}(\cdot,\a_{\text{BAWGNMA}})$) be the density transformation
operator corresponding to a message from user $1$ to user $2$ (user
$2$ to user $1$) via the function node. More precisely,
\begin{align*}
  f_{1\to 2}(\a,\a_{\text{BAWGNMA}}) &\triangleq \sum_{x\in\mathcal{X}}p(x)f_{12}(\a(\pi_1(x)u),\nu_x(u))\\
  f_{2\to 1}(\b,\a_{\text{BAWGNMA}}) &\triangleq \sum_{x\in\mathcal{X}}p(x)f_{21}(\b(\pi_2(x)u),\nu_x(u)),
\end{align*}
where $f_{12}(\cdot,\cdot)$ and $f_{21}(\cdot,\cdot)$ are density
transformation operators corresponding to~\eqref{eq:fun_12} and
\eqref{eq:fun_21}. Here, $\a(u)$ (respectively $\b(u)$) is the density
of the messages $m_{i,v\to f}^{[1]}$ ($m_{i,v\to f}^{[2]}$). These
operators can be computed numerically for discretized densities
following the procedure outlined in \cite{Chung-comlett01}. Let
$\a_{\ell}$ and $\b_{\ell}$ denote the $L$-density\footnote{Since the
  transmission alphabet is $\{\pm 1\}$, the densities are conditioned
  assuming the transmission of a $+1$.}  of the messages emanating
from the variable nodes to the check nodes at iteration $\ell$,
corresponding to codes $1$ and $2$. Then, the DE equations for the
joint decoder are given by
\begin{align*}
  \a_{\ell+1} &=
  f_{2\to 1}\Bigl(L\left(\rho(\b_{\ell})\right),\a_{\text{BAWGNMA}}\Bigr)\oast
  \lambda(\rho(\a_{\ell}))\\
  \b_{\ell+1} &=
  f_{1\to 2}\Bigl(L\left(\rho(\a_{\ell})\right),\a_{\text{BAWGNMA}}\Bigr)\oast
  \lambda(\rho(\b_{\ell})),
\end{align*}
where $\lambda(\a)=\sum_i\lambda_i\a^{\oast(i-1)}$,
$L(\a)=\sum_iL_i\a^{\oast(i-1)}$, and
$\rho(\a)=\sum_i\rho_i\a^{\boxast(i-1)}$. Here, $L(\rho(\cdot))$ is
the density of the messages from the variable node to the function
node. These equations accurately represent the evolution of densities
at the decoder due to the symmetry of the variable and check node
operations. The fixed points of density evolution are the triples
$(\a_{\text{BAWGNMA}},\a,\b)$ which satisfy
\begin{align}
  \label{eq:de_fp}
  \a &=
  f_{2\to 1}\Bigl(L\left(\rho(\b)\right),\a_{\text{BAWGNMA}}\Bigr)\oast
  \lambda(\rho(\a))\notag\\
  \b &=
  f_{1\to 2}\Bigl(L\left(\rho(\a)\right),\a_{\text{BAWGNMA}}\Bigr)\oast
  \lambda(\rho(\b)).
\end{align}
Let $\Delta_{+\infty}$ denote the delta function at $+\infty$ and let
$\hb=A\cdot\ha$, for some $A\in[0,+\infty]$. The BP threshold is defined
by
\begin{align*}
  h^{[1],\text{BP}}(\lambda,\rho,A) &= \sup\bigl\{\ha :\text{ The fixed point equation}\\
  & ~\eqref{eq:de_fp}\text{ has a solution
  }(\a,\b)\neq(\Delta_{+\infty},\Delta_{+\infty})\bigr\},\\
  h^{[2],\text{BP}}(\lambda,\rho,A) &= Ah^{[1],\text{BP}}(\lambda,\rho,A).
\end{align*}
The set of all points $(h^{[1]'},Ah^{[1]'})\ni h^{[1]'}\geq
h^{[1],\text{BP}}(\lambda,\rho,A)$ is called the BP-ACPR and its
boundary is called the DE boundary.

\begin{figure}
  \centering
  \begin{tikzpicture}[scale=0.53,>=stealth]

\draw[rounded corners,fill=brown,opacity=0.4] (-1,2) rectangle +(13,1);
\draw (5.5,2.5) node {permutation $\pi_1$};
\draw[rounded corners,fill=brown,opacity=0.4] (-1,9) rectangle +(13,1);
\draw (5.5,9.5) node {permutation $\pi_2$};

\foreach \x in {0,2,4,6,11}
{
  \filldraw[black] (\x,0.9)+(-4pt,-4pt) rectangle +(4pt,4pt);
  \draw (\x,0.9)+(0pt,-4pt) -- (\x,2);
  \draw (\x,0.9)+(0pt,-4pt) -- ([xshift = 3pt]\x,2);
  \draw (\x,0.9)+(0pt,-4pt) -- ([xshift = 6pt]\x,2);
  \draw (\x,0.9)+(0pt,-4pt) -- ([xshift = -3pt]\x,2);
  \draw (\x,0.9)+(0pt,-4pt) -- ([xshift = -6pt]\x,2);

  \filldraw[black] (\x,11)+(-4pt,-4pt) rectangle +(4pt,4pt);
  \draw (\x,11) -- (\x,10);
  \draw (\x,11) -- ([xshift = 6pt]\x,10);
  \draw (\x,11) -- ([xshift = 3pt]\x,10);
  \draw (\x,11) -- ([xshift = -3pt]\x,10);
  \draw (\x,11) -- ([xshift = -6pt]\x,10);
}

\foreach \x in {-0.5,1.5,3.5,5.5,11.5} {
  \shade[ball color=blue] (\x,4) circle (4pt);
\begin{pgfonlayer}{background}
  \draw (\x,4)+(0cm,-2pt) -- ([xshift=0cm]\x,3);
  \draw (\x,4)+(0cm,-2pt) -- ([xshift=-0.2cm]\x,3);
  \draw (\x,4)+(0cm,-2pt) -- ([xshift=0.2cm]\x,3);
\end{pgfonlayer}

  \shade[ball color=blue] (\x,8) circle (4pt);
\begin{pgfonlayer}{background}
  \draw (\x,8)+(0cm,2pt) -- ([xshift=0cm]\x,9);
  \draw (\x,8)+(0cm,2pt) -- ([xshift=-0.2cm]\x,9);
  \draw (\x,8)+(0cm,2pt) -- ([xshift=0.2cm]\x,9);
\end{pgfonlayer}
}

\foreach \x in {-0.5,1.5,3.5,5.5,11.5} {
  \node[diamond,draw=black,fill=red,inner sep=0pt,minimum size=6pt] at (\x,6)
  {};
\begin{pgfonlayer}{background}
  \draw[->] (\x,6)+(-1,0) -- ([xshift=-4pt]\x,6);
\end{pgfonlayer}
}

\begin{pgfonlayer}{background}
\draw (-1,4)+(0.5cm,2pt) -- ([xshift=0.5cm]-1,7.925);
\draw (1,4)+(0.5cm,2pt) -- (1.5,7.925);
\draw (3,4)+(0.5cm,2pt) -- (3.5,7.925);
\draw (5,4)+(0.5cm,2pt) -- (5.5,7.925);
\draw (11,4)+(0.5cm,2pt) -- (11.5,7.925);
\end{pgfonlayer}

\foreach \x in {7.5,8,8.5,...,9.5} {
  \foreach \y in {1,4,6,8,11} {
    \filldraw (\x,\y) circle (1pt);
  }
}

\draw (-2,1.5) node {$\rho(x)$};
\draw (-2,3.5) node {$\lambda(x)$};
\draw (-2,10.5) node {$\rho(x)$};
\draw (-2,8.5) node {$\lambda(x)$};



\draw (12.7,6) node[black] {$f(\cdot,\cdot)$};
\draw[->,thick,gray] (6,5.5) -- (6,7);
\node (f) at (6,5) {$f_{12}$};
\draw[->,thick,gray] (4,6.5) -- (4,5);
\node (f) at (4,7) {$f_{21}$};

\draw[->,thick,gray] (12.5,3.5) -- (12.5,1.5)  node[black,midway,right=2pt] {$\mathsf{a}_{\ell}$};
\draw[->,thick,gray] (12.5,8.5) -- (12.5,10.5) node[black,midway,right=2pt] {$\mathsf{b}_{\ell}$};
\end{tikzpicture}

  \caption{Tanner graph of the joint decoder. The variable nodes of
    each code are connected through function nodes, which receives the
  channel outputs. The joint decoder iterates by passing messages
  between the component decoders.}
  \label{fig:jd_tanner}
\end{figure}
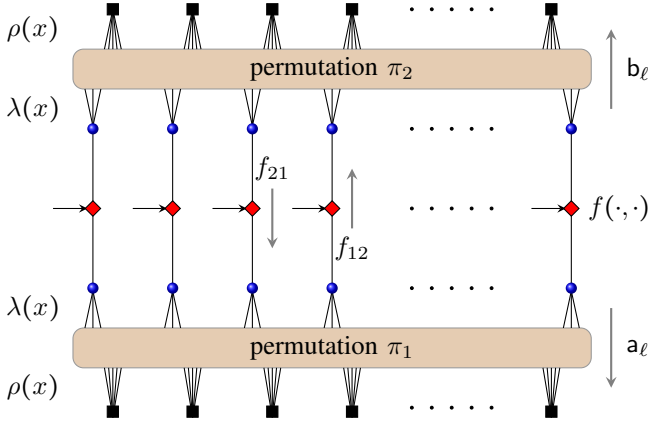

\subsection{GEXIT curves}
\label{sec:gexit-curves}
The following approach can be applied to any binary-input MAC
characterized by a single parameter, whose density is differentiable
and degraded with respect to that parameter. A through discussion of
channel degradation can be found in \cite[p. 204]{RU-2008}. The
binary-input Gaussian MAC defined by \eqref{eq:gmac} satisfies these
conditions for the parameter $\alpha$, with $h_1 = \alpha$ and $h_2 =
A\alpha$, for any fixed $A\in[0,+\infty)$. We note that any
differentiable parameterization would suffice for the following
discussion. 

Now, suppose that the $i$th bit is transmitted through a channel with
parameter $\alpha_{i}$ and that each $\alpha_i$ is a differentiable
function of some parameter $\alpha$. We omit the dependence of the
channel on $A$ throughout this section for clarity of notation. The
GEXIT curve is defined by
\begin{align*}
  \gexit(\alpha) \triangleq \frac 1n
  \sum_{i=1}^{n}\underbrace{\frac{\partial}{\partial\alpha_{i}}\entropy{\X|\Y(\alpha)}}_{\triangleq
  \gexit_{i}(\alpha_i)}\frac{\partial\alpha_i}{\partial\alpha}.
\end{align*}
A more convenient expression for the GEXIT curve can be derived
following the procedure given in \cite{Measson-2006} for non-binary
codes. Let $\mathbf{y}_{\sim i} = \mathbf{y}\backslash y_i$,
$\phi_{i}(\mathbf{y}_{\sim i}) = \{p_{X_i|\Y_{\sim
    i}}(x|\mathbf{y}_{\sim i}),x\in\mathcal{X}\}$ and
$\Phi_i\triangleq\phi_i(\Y_{\sim i})$ be the corresponding random
variable. Note that $\phi_i(\cdot)$ is the extrinsic MAP estimator of
$X_i$\footnote{To see this, write
  \begin{align*}
    p_{\Y_{\sim i}|X_i}\!(\mathbf{y}_{\sim i}|x_i)\!=\!
    \frac{p_{X_i|\Y_{\sim i}}\!(x_i|\mathbf{y}_{\sim
        i})}{p_{X_i}(x_i)}p_{\Y_{\sim i}}\!(\mathbf{y}_{\sim i})\!=\!
    \frac{\phi_i\cdot e[x_i]}{p_{X_i}(x_i)}p_{\Y_{\sim
        i}}\!(\mathbf{y}_{\sim i}),
  \end{align*}
where $e[x_i]$ is the standard basis vector with a $1$ in the
$x_i$-th coordinate and use the result in \cite[p. 29]{RU-2008}.}. 

\begin{lemma}
  Suppose that all bits are transmitted through channel with
  parameter $\alpha$. Then, the $i$th GEXIT function is given by
  \begin{align*}
    \gexit_i(\alpha) = \sum_{x\in\mathcal{X}}p(x)\int\limits_{\mathbf{u}}\a_{x,i}(\mathbf{u})\kappa_{x}(\mathbf{u})\text{d}\mathbf{u},
  \end{align*}
  where $\a_{x,i}(\mathbf{u})$ is the distribution of
  $\Phi_i$ given $X_i = x$ and the GEXIT kernel is given by
  \begin{align}
    \label{eq:kernel}
    \kappa_{x}(\mathbf{u}) = 
\int\frac{\partial}{\partial\alpha} p(y|x)\!
    \log_2\frac{\sum_{x'}u[x']p(y|x')}
    {u[x]p(y|x)}\text{d}y,
  \end{align}
  where $u[j]$ denotes the $j$th component of $\mathbf{u}$.
\end{lemma}
\begin{proof}
Suppose that each bit is transmitted through a channel with parameter
$\alpha_i$ and consider the term
\begin{align*}
    \entropy{\X|\Y} &= \entropy{X_i|\Y} + \entropy{\X_{\sim i}|X_i,\Y}\\
    &= \entropy{X_{i}|Y_{i},\Phi_{i}} + \entropy{\X_{\sim i} | X_{i},\Y_{\sim i}}.
\end{align*}
Note that the second term of the decomposition does not depend on the
channel at position $i$. So, we get
\begin{align*}
  \gexit_{i}(\alpha_i) = \frac{\d}{\d\alpha_{i}}\entropy{X_{i}|Y_{i},\Phi_{i}}.
\end{align*}
We have,
\begin{align*}
  &\entropy{X_{i}|Y_{i},\Phi_{i}} = -\!\!\iint\limits_{y,\phi}\!\!\sum_{x}
  \!p(x,y,\phi)\!\log_2\!\frac{p(x,y,\phi)}{\displaystyle\sum_{x'}p(x',y,\phi)}\text{d}y\text{d}\phi\\
  &= \sum_{x}p(x)\!\!\int\limits_{\phi}\!\!p(\phi|x)\! \left(\!\int\limits_{y}\!\!p(y|x)\!
    \log_2\frac{\displaystyle\sum_{x'}p(x'|\phi)p(y|x')}
    {p(x|\phi)p(y|x)}\text{d}y\!\right)\!\text{d}\phi,
\end{align*}
which follows from the fact that
\begin{align*}
  p_{X_{i},Y_{i},\Phi_{i}}(x,y,\phi) &= p(y|x)p(\phi|x)p(x),
\end{align*}
since $Y_i\to X_i\to \Phi_i$. Taking the derivative and noting that
$p(x_i|\phi_i) = p(x_i|\mathbf{y}_{\sim i})$, we obtain\footnote{The
  terms obtained by differentiating with respect to the channel inside
  the $\log$ vanish.}
\begin{align*}
  \gexit_{i}(\alpha_i) &=
  \sum_{x}p(x)\int\limits_{\phi}p(\phi|x)
  \biggl(\int\limits_{y}\frac{\partial}{\partial\alpha}p(y|x) \\
    &\phantom{\sum_xp(x)}\log_2\frac{\displaystyle\sum_{x'}p(x'|\phi)p(y|x')}
    {p(x|\phi)p(y|x)}\text{d}y\!\biggr)\text{d}\phi \\
    &= \sum_xp(x)\int\limits_{\mathbf{u}}\a_{x}(\mathbf{u})\kappa_{x}(\mathbf{u})\text{d}\mathbf{u},
\end{align*}
and the result follows by setting $\alpha_i=\alpha$.
\end{proof}

The GEXIT function is hard to compute and hence we use the BP-GEXIT
function instead \cite{Measson-it09}. The BP-GEXIT function is
obtained by replacing the MAP extrinsic estimator with the
corresponding BP estimator. Let $\Phi_i^{\text{BP},\ell,n}$ denote the
BP extrinsic estimate of $X_i$ after $\ell$ iterations of the joint
decoder. The BP extrinsic estimate is computed using the computation
graph of depth $\ell$ for function node $i$. Define the BP-GEXIT
function at the $\ell$th iteration $\gexit^{\text{BP},\ell,n}(\alpha)$
in a similar manner to \cite{Measson-it09} (taking an expectation over
all possible computation graphs) and the asymptotic BP-GEXIT function
$\gexit^{\text{BP}}(\alpha) = \lim_{\ell\to\infty} \lim_{n\to\infty}
\gexit^{\text{BP},\ell,n}(\alpha)$. For fixed $\ell$, in the limit of
$n\to\infty$, the computation graph becomes tree-like because the
computation graphs of the two variable nodes (which themselves become
tree-like) connected to the function node do not overlap with high
probability. The extrinsic estimate of $X_i$ can then be computed via
the extrinsic estimates of $X_i^{[1]}$ and $X_i^{[2]}$. The asymptotic
BP-GEXIT function can be computed through the fixed points of density
evolution $(\a_{\text{BAWGNMA}(\alpha)},\a,\b)$ which satisfy
\eqref{eq:de_fp} and is discussed in the following Lemma.

\begin{lemma}
  Consider transmission over the channel
  $\a_{\text{BAWGNMA}}(\alpha)$ and let
    $(\a_{\text{BAWGNMA}(\alpha)},\a,\b)$ be a fixed point of
    DE. Define the BP-GEXIT value of the fixed point by
\begin{align*}
  &\G^{\text{BP}}\!(\a_{\text{BAWGNMA}},\a,\b) \!\triangleq\!\!
  \sum_{x\in\mathcal{X}}\!p(x)\!\!\! \int\!\!\! \mathsf{F}_x[\a,\b](u,v) \kappa_x(u,v)
  \text{d}u\text{d}v. 
\end{align*}
The GEXIT kernel $\kappa_x(\cdot,\cdot)$ is defined as in
\eqref{eq:kernel} and the operator $\mathsf{F}_x[\cdot,\cdot]$
computes the density of the extrinsic BP estimate $\Phi^{\text{BP}}$
given $X=x$. The BP-GEXIT curve $\mathsf{g}^{\text{BP}}(\alpha)$ is
given in parametric form by
$(\alpha,\G(\a_{\text{BAWGNMA}(\alpha)},\a,\b))$.
\end{lemma}
\begin{proof}
  Let $\Phi^{\text{BP}} = \mathbf{u}$. Then,
  \begin{align*}
   \mathbf{u} &= \phi(\mathbf{y}_{\sim i}) = \left\{p(x_i|\mathbf{y}_{\sim i}),x_i\in\mathcal{X}\right\} \\ 
    &= \{p(\pi_1(x_i)|\mathbf{y}_{\sim i})\cdot p(\pi_2(x_i)|\mathbf{y}_{\sim i}),x_i\in\mathcal{X}\}.
  \end{align*}
  If we define
  \begin{align*}
 u \triangleq \log\frac{p(X_i^{[1]} = +1|\mathbf{y}_{\sim i})}{p(X_i^{[1]} = -1
  | \mathbf{y}_{\sim i})},  v \triangleq \log\frac{p(X_i^{[2]} = +1|\mathbf{y}_{\sim i})}{p(X_i^{[2]} = -1
  | \mathbf{y}_{\sim i})},   
  \end{align*}
  then
  \begin{align}
    \label{eq:map}
    \mathbf{u}\! &=\! \biggl(\!\frac{\e^u}{1 + \e^u}\!\frac{\e^v}{1 + \e^v},
    \frac{\e^u}{1 + \e^u}\!\frac{1}{1 + \e^v}, \frac{1}{1 +
      \e^u}\!\frac{\e^v}{1 + \e^v}, \frac{1}{1 + \e^u}\!\frac{1}{1 +
      \e^v}\!\biggr)\notag \\
    &\triangleq f(u,v). 
  \end{align}
  Let $\a(u)$ denote the density of $U$ conditioned on $X_i^{[1]} =
  +1$ and $\b(v)$ be the density of $V$ conditioned on $X_i^{[2]} =
  +1$. Then, $\a(-u)$ is the density of $U$ conditioned on $X_i^{[1]}
  = -1$ and $\b(-v)$ is the density of $V$ conditioned on $X_i^{[2]} =
  -1$. In the limit $n\to\infty$ and taking expectation these
  densities are given by the fixed point
  $(\a_{\text{BAWGNMA}(\alpha)},\a,\b)$. Let
  $\mathsf{F}_x[\a,\b](u,v)$ be the density of $\Phi_i^{\text{BP}}$
  conditioned on $(X_i^{[1]} = \pi_1(x), X_i^{[2]} = \pi_2(x))$. Then,
  \begin{align*}
    \mathsf{F}_x[\a,\b](u,v) = \a\left(\pi_1(x)u\right)\b\left(\pi_2(x)v\right).
  \end{align*}
  For example $\mathsf{F}_0[\a,\b](u,v) = \a(u)\b(v),
  \mathsf{F}_1[\a,\b](u,v) = \a(u)\b(-v)$ and so on. The result
  follows by the definition of the GEXIT curve. The kernels
  $\kappa_x(u,v)$ are defined in the sense of \eqref{eq:map}.
\end{proof}

It can be shown that the BP-GEXIT function is an upper bound on the
GEXIT function (see the discussion in \cite[p. 206]{RU-2008}). The
BP-GEXIT curve for the LDPC$(3,6)$ ensemble is shown in
Fig.~\ref{fig:bp-gexit}, for $A=1$.

\begin{figure}
  \centering
  \begin{tikzpicture}[scale=0.5,>=stealth]

\definecolor{mycolor1}{rgb}{0,1,1}
\definecolor{mycolor2}{rgb}{1,0,1}

\begin{axis}[
scale only axis,
width=5in,
height=5in,
xmin=0, xmax=2,
ymin=-1.4, ymax=0,
xlabel=\mbox{\Large $\alpha$},
ylabel=\mbox{\Large $\mathsf{g}^{\text{BP}}$},
xmajorgrids,
ymajorgrids]

\addplot [
color=blue,
line width=1.25pt,
solid
]
coordinates{
 (0,-5.42634e-15)
 (0.01,-0.0288481)
 (0.02,-0.0576617)
 (0.03,-0.0864062)
 (0.04,-0.115047)
 (0.05,-0.143552)
 (0.06,-0.171886)
 (0.07,-0.200017)
 (0.08,-0.227914)
 (0.09,-0.255545)
 (0.1,-0.282881)
 (0.11,-0.309893)
 (0.12,-0.336553)
 (0.13,-0.362835)
 (0.14,-0.388714)
 (0.15,-0.414166)
 (0.16,-0.439169)
 (0.17,-0.463702)
 (0.18,-0.487746)
 (0.19,-0.511283)
 (0.2,-0.534298)
 (0.21,-0.556775)
 (0.22,-0.578701)
 (0.23,-0.600064)
 (0.24,-0.620855)
 (0.25,-0.641065)
 (0.26,-0.660686)
 (0.27,-0.679713)
 (0.28,-0.698141)
 (0.29,-0.715968)
 (0.3,-0.73319)
 (0.31,-0.749808)
 (0.32,-0.765823)
 (0.33,-0.781234)
 (0.34,-0.796047)
 (0.35,-0.810264)
 (0.36,-0.823888)
 (0.37,-0.836929)
 (0.38,-0.849388)
 (0.39,-0.861277)
 (0.4,-0.872598)
 (0.41,-0.883366)
 (0.42,-0.893589)
 (0.43,-0.903272)
 (0.44,-0.912423)
 (0.45,-0.921064)
 (0.46,-0.929198)
 (0.47,-0.936837)
 (0.48,-0.943992)
 (0.49,-0.950676)
 (0.5,-0.956902)
 (0.51,-0.962675)
 (0.52,-0.96802)
 (0.53,-0.972938)
 (0.54,-0.977447)
 (0.55,-0.981552)
 (0.56,-0.985281)
 (0.57,-0.988637)
 (0.58,-0.991625)
 (0.59,-0.994266)
 (0.6,-0.996567)
 (0.61,-0.998535)
 (0.62,-1.00019)
 (0.63,-1.00154)
 (0.64,-1.00258)
 (0.65,-1.00335)
 (0.66,-1.00386)
 (0.67,-1.0041)
 (0.68,-1.00408)
 (0.69,-1.00382)
 (0.7,-1.00332)
 (0.71,-1.0026)
 (0.72,-1.00166)
 (0.73,-1.00051)
 (0.74,-0.999149)
 (0.75,-0.997601)
 (0.76,-0.995861)
 (0.77,-0.99394)
 (0.78,-0.991841)
 (0.79,-0.989572)
 (0.8,-0.98714)
 (0.81,-0.984549)
 (0.82,-0.981801)
 (0.83,-0.978908)
 (0.84,-0.975869)
 (0.85,-0.972684)
 (0.86,-0.969363)
 (0.87,-0.965917)
 (0.88,-0.962338)
 (0.89,-0.958635)
 (0.9,-0.954805)
 (0.91,-0.950853)
 (0.92,-0.946785)
 (0.93,-0.942603)
 (0.94,-0.938315)
 (0.95,-0.933922)
 (0.96,-0.929425)
 (0.97,-0.924805)
 (0.98,-0.920068)
 (0.99,-0.915247)
 (1,-0.910315)
 (1.01,-0.90528)
 (1.02,-0.900199)
 (1.03,-0.894993)
 (1.04,-0.889678)
 (1.05,-0.884259)
 (1.06,-0.878783)
 (1.07,-0.87322)
 (1.08,-0.86755)
 (1.09,-0.861828)
 (1.1,-0.85597)
 (1.11,-0.850072)
 (1.12,-0.844073)
 (1.13,-0.838029)
 (1.14,-0.831877)
 (1.15,-0.825659)
 (1.16,-0.81937)
 (1.17,-0.813015)
 (1.18,-0.806605)
 (1.19,-0.800103)
 (1.2,-0.79354)
 (1.21,-0.786886)
 (1.22,-0.780155)
 (1.23,-0.773417)
 (1.24,-0.766598)
 (1.25,-0.759702)
 (1.26,-0.752815)
 (1.27,-0.745878)
 (1.28,-0.738858)
 (1.29,-0.731805)
 (1.3,-0.724716)
 (1.31,-0.71758)
 (1.32,-0.710319)
 (1.33,-0.70309)
 (1.34,-0.695805)
 (1.35,-0.688533)
 (1.36,-0.681237)
 (1.37,-0.673863)
 (1.38,-0.666452)
 (1.39,-0.659047)
 (1.4,-0.65163)
 (1.41,-0.644136)
 (1.42,-0.636657)
 (1.43,-0.62914)
 (1.44,-0.621651)
 (1.45,-0.61409)
 (1.46,-0.606589)
 (1.47,-0.598988)
 (1.48,-0.591449)
 (1.49,-0.583859)
 (1.5,-0.576268)
 (1.51,-0.568711)
 (1.52,-0.561142)
 (1.53,-0.553569)
 (1.54,-0.545893)
 (1.55,-0.53833)
 (1.56,-0.530777)
 (1.57,-0.523132)
 (1.58,-0.515497)
 (1.59,-0.507818)
 (1.6,-0.500099)
 (1.61,-0.492378)
 (1.62,-0.484632)
 (1.63,-0.476879)
 (1.64,-0.469073)
 (1.65,-0.461094)
 (1.66,-0.452923)
 (1.67,-0.444385)
 (1.68,-0.435389)
 (1.69,-0.424852)
 (1.69,-2.96196e-06)
 (1.7,-2.96196e-06)
 (1.71,-2.97647e-06)
 (1.72,-2.99078e-06)
 (1.73,-3.00487e-06)
 (1.74,-3.01874e-06)
 (1.75,-3.03239e-06)
 (1.76,-3.04579e-06)
 (1.77,-3.05894e-06)
 (1.78,-3.07183e-06)
 (1.79,-3.08444e-06)
 (1.8,-3.09678e-06)
 (1.81,-3.10882e-06)
 (1.82,-3.12056e-06)
 (1.83,-3.13199e-06)
 (1.84,-3.1431e-06)
 (1.85,-3.15387e-06)
 (1.86,-3.16429e-06)
 (1.87,-3.17436e-06)
 (1.88,-3.18406e-06)
 (1.89,-3.19339e-06)
 (1.9,-3.20233e-06)
 (1.91,-3.21087e-06)
 (1.92,-3.21901e-06)
 (1.93,-3.22672e-06)
 (1.94,-3.23401e-06)
 (1.95,-3.24086e-06)
 (1.96,-3.24726e-06)
 (1.97,-3.2532e-06)
 (1.98,-3.25868e-06)
 (1.99,-3.26369e-06)
 (2,-3.26821e-06)};

\addplot [
fill=gray,fill opacity=0.2,
line width=1.25pt,
solid
]
coordinates{
 (0,-5.42634e-15)
 (0.01,-0.0288481)
 (0.02,-0.0576617)
 (0.03,-0.0864062)
 (0.04,-0.115047)
 (0.05,-0.143552)
 (0.06,-0.171886)
 (0.07,-0.200017)
 (0.08,-0.227914)
 (0.09,-0.255545)
 (0.1,-0.282881)
 (0.11,-0.309893)
 (0.12,-0.336553)
 (0.13,-0.362835)
 (0.14,-0.388714)
 (0.15,-0.414166)
 (0.16,-0.439169)
 (0.17,-0.463702)
 (0.18,-0.487746)
 (0.19,-0.511283)
 (0.2,-0.534298)
 (0.21,-0.556775)
 (0.22,-0.578701)
 (0.23,-0.600064)
 (0.24,-0.620855)
 (0.25,-0.641065)
 (0.26,-0.660686)
 (0.27,-0.679713)
 (0.28,-0.698141)
 (0.29,-0.715968)
 (0.3,-0.73319)
 (0.31,-0.749808)
 (0.32,-0.765823)
 (0.33,-0.781234)
 (0.34,-0.796047)
 (0.35,-0.810264)
 (0.36,-0.823888)
 (0.37,-0.836929)
 (0.38,-0.849388)
 (0.39,-0.861277)
 (0.4,-0.872598)
 (0.41,-0.883366)
 (0.42,-0.893589)
 (0.43,-0.903272)
 (0.44,-0.912423)
 (0.45,-0.921064)
 (0.46,-0.929198)
 (0.47,-0.936837)
 (0.48,-0.943992)
 (0.49,-0.950676)
 (0.5,-0.956902)
 (0.51,-0.962675)
 (0.52,-0.96802)
 (0.53,-0.972938)
 (0.54,-0.977447)
 (0.55,-0.981552)
 (0.56,-0.985281)
 (0.57,-0.988637)
 (0.58,-0.991625)
 (0.59,-0.994266)
 (0.6,-0.996567)
 (0.61,-0.998535)
 (0.62,-1.00019)
 (0.63,-1.00154)
 (0.64,-1.00258)
 (0.65,-1.00335)
 (0.66,-1.00386)
 (0.67,-1.0041)
 (0.68,-1.00408)
 (0.69,-1.00382)
 (0.7,-1.00332)
 (0.71,-1.0026)
 (0.72,-1.00166)
 (0.73,-1.00051)
 (0.74,-0.999149)
 (0.75,-0.997601)
 (0.76,-0.995861)
 (0.77,-0.99394)
 (0.78,-0.991841)
 (0.79,-0.989572)
 (0.8,-0.98714)
 (0.81,-0.984549)
 (0.82,-0.981801)
 (0.83,-0.978908)
 (0.84,-0.975869)
 (0.85,-0.972684)
 (0.86,-0.969363)
 (0.87,-0.965917)
 (0.88,-0.962338)
 (0.89,-0.958635)
 (0.9,-0.954805)
 (0.91,-0.950853)
 (0.92,-0.946785)
 (0.93,-0.942603)
 (0.94,-0.938315)
 (0.95,-0.933922)
 (0.96,-0.929425)
 (0.97,-0.924805)
 (0.98,-0.920068)
 (0.99,-0.915247)
 (1,-0.910315)
 (1.01,-0.90528)
 (1.02,-0.900199)
 (1.03,-0.894993)
 (1.04,-0.889678)
 (1.05,-0.884259)
 (1.06,-0.878783)
 (1.07,-0.87322)
 (1.08,-0.86755)
 (1.09,-0.861828)
 (1.1,-0.85597)
 (1.11,-0.850072)
 (1.12,-0.844073)
 (1.13,-0.838029)
 (1.14,-0.831877)
 (1.15,-0.825659)
 (1.16,-0.81937)
 (1.17,-0.813015)
 (1.18,-0.806605)
 (1.19,-0.800103)
 (1.2,-0.79354)
 (1.21,-0.786886)
 (1.22,-0.780155)
 (1.23,-0.773417)
 (1.24,-0.766598)
 (1.25,-0.759702)
 (1.26,-0.752815)
 (1.2629,-0.7508)
 (1.2629,0)
}|- (axis cs:0,0) -- cycle;

\node[rotate=90] (map) at (axis cs:1.315,-0.16) {\LARGE{$\bar{\alpha} = 1.2629$}};
\node[blue,rotate=90] (bp) at (axis cs:1.745,-0.16) {\LARGE{$\alpha^{\text{BP}} = 1.69$}};
\node[black] (area) at (axis cs:0.7,-0.35) 
{\Huge{$\int\limits_{\bar{\alpha}}^0\mathsf{g}^{\text{BP}}(\alpha) = 1$}};

\end{axis}

\end{tikzpicture}

  \caption{BP-GEXIT curve and an upper bound on the MAP threshold
    (computed using the area theorem), for $A=1$, of the $(3,6)$
    regular LDPC ensemble. GEXIT curves in literature are typically
    parameterized by the channel entropy and the channels get
    \emph{worse} as the entropy increases. However, the channel gains
    are a natural parameterization for this problem and the channel
    gets \emph{better} by increasing the channel gains. So the GEXIT
    values are negative for this parameterization.}
  \label{fig:bp-gexit}
\end{figure}
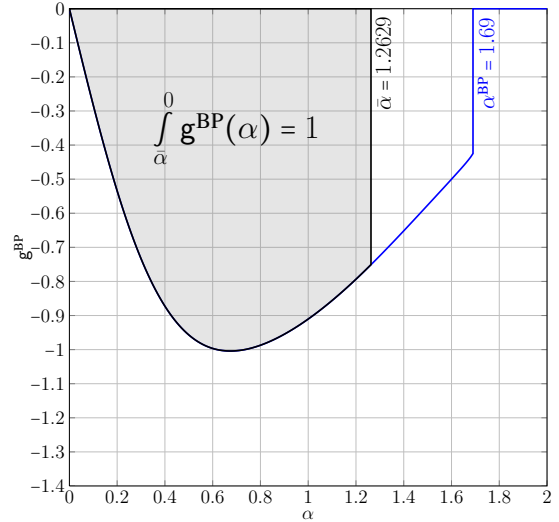

\subsection{The MAP Boundary}
\label{sec:map-threshold}
Let $\ha _{i}(\alpha) = \alpha$ and $\hb _{i}(\alpha) =
A\alpha$, for all $i=1,\ldots,n$. Consider transmission using codes
from the ensemble LDPC$(n,\lambda,\rho)$. For a fixed $A$, we define
the MAP threshold as
\begin{align*}
  \alpha^{\text{MAP}}(A) = \inf\left\{\alpha:
  \liminf_{n\to\infty}\frac 1n\mathbb{E}[H(\X|\Y(\alpha,A))] > 0\right\},
\end{align*}
where the expectation is taken over all codes in the ensemble. The set
of all points $(\alpha',A)\ni \alpha'\geq\alpha^{\text{MAP}}(A)$ form the
MAP-ACPR and its boundary is called the MAP boundary. By
definition of the GEXIT function, this gives
\begin{align*}
  \int^{0}_{+\infty}\gexit(\alpha)\d\alpha &= \frac 1n
  \int^{0}_{+\infty}\frac{\d\entropy{\X|\Y(\alpha)}}{\d\alpha}\d\alpha \\
  &= \frac 1n \entropy{\X|\Y(0)} \\
  &= \frac{2k}{n}.
\end{align*}
The above equation gives us a procedure to compute the MAP threshold,
using the GEXIT curve. Let $\bar{\alpha}$ denote the largest positive
number such that
\begin{align*}
  \int_{\bar{\alpha}}^{0}\gexit^{\text{BP}}(\alpha)\d\alpha &= 2\rate(\lambda,\rho),
\end{align*}
where $\rate(\lambda,\rho)$ is the design rate of the ensemble
LDPC$(\lambda,\rho)$. Then the MAP threshold $\alpha^{\text{MAP}} \leq
\bar{\alpha}$. The upper bound on the MAP threshold is shown in
Fig.~\ref{fig:bp-gexit} for the $(3,6)$ LDPC ensemble, for
$A=1$. Using this procedure, we can compute an outer bound to the MAP
boundary by considering different values of $A$.


\section{Spatial Coupling}
\label{sec:spatial-coupling}
Spatial coupling is best described by the $(l,r,L)$ ensemble through a
protograph \cite{Kudekar-it11,Sridharan-aller04}. The protograph
structure at the joint decoder is shown in Fig.~\ref{fig:protograph}
for a LDPC$(3,6)$ base code. The protograph is generated as follows:
Consider the protograph of a $(3,6)$ regular LDPC code. It has two
variable nodes of degree $3$ and one check node of degree $6$. Connect
both the variable nodes to the variable nodes of another protograph
via function nodes. The resulting protograph represents the joint
decoder when both users are using $(3,6)$ regular LDPC codes for
transmission over the $2$-user binary-input Gaussian MAC. Place $2L+1$
protographs at positions $-L,\cdots,L$. Each of the $3$ edges of a
variable node at position $i$ is connected to exactly one check node
at position $i-1,i,i+1$, for each user.

\begin{figure}[htb]
  \centering
  \begin{tikzpicture}[scale=0.35]
\useasboundingbox (0,-7) rectangle (24,13);
\begin{scope}[yshift=-15cm] 
  \foreach \x in {0,3,6,9,12,15,18,21,24} {
    \draw[fill=black] (\x,10)+(-5pt,-5pt) rectangle +(5pt,5pt)
    node[left=4pt,below=7pt,outer sep=0pt,inner sep=0pt] (c\x) {};
    \node[outer sep=0pt,inner sep=0pt] (cb\x) at (\x,10) {};
  }
  \foreach \x in {2,4,5,7,8,10,11,13,14,16,17,19,20,22} {
    \foreach \y in {12} {
      \shade[ball color=blue] (\x,\y) circle (6pt) node[outer
      sep=0pt,inner sep=0pt] (vb\x\y) {};
    }	
  }
  \begin{pgfonlayer}{foreground}
    \foreach \x in {14.5,15,15.5} {
      \foreach \y in {10,12} {
        \draw[fill=black] (\x,\y) circle (0.5pt);
      }
    }
  \end{pgfonlayer}
  \draw[white,fill=white] (13.75,8) rectangle +(2.5,4.5);
  \begin{pgfonlayer}{background}
    \foreach \x/\y/\z in {2/12/0,2/12/3,2/12/6,4/12/0,4/12/3,4/12/6,
    5/12/3,5/12/6,5/12/9,7/12/3,7/12/6,7/12/9,
    8/12/6,8/12/9,8/12/12,10/12/6,10/12/9,10/12/12,
    11/12/9,11/12/12,11/12/15,13/12/9,13/12/12,13/12/15,
    14/12/12,14/12/15,14/12/18,16/12/12,16/12/15,16/12/18,
    17/12/15,17/12/18,17/12/21,19/12/15,19/12/18,19/12/21,
    20/12/18,20/12/21,20/12/24,22/12/18,22/12/21,22/12/24} {
      \draw (vb\x\y) -- (cb\z);
    }
  \end{pgfonlayer}
\end{scope}

\foreach \x in {0,3,6,9,12,15,18,21,24} {
  \draw[fill=black] (\x,2)+(-5pt,-5pt) rectangle +(5pt,5pt)
  node[left=4pt,below=7pt,outer sep=0pt,inner sep=0pt] (c\x) {};
  \node[outer sep=0pt,inner sep=0pt] (c\x) at (\x,2) {};
}
\foreach \x in {2,4,5,7,8,10,11,13,14,16,17,19,20,22} {
  \foreach \y in {0} {
    \shade[ball color=blue] (\x,\y) circle (6pt) node[outer
    sep=0pt,inner sep=0pt] (v\x\y) {};
  }	
}
\begin{pgfonlayer}{foreground}
  \foreach \x in {14.5,15,15.5} {
    \foreach \y in {-1.5,0,2} {
      \draw[fill=black] (\x,\y) circle (0.5pt);
    }
  }
\end{pgfonlayer}
\draw[white,fill=white] (13.75,-0.5) rectangle +(2.5,4.5);
      
\begin{pgfonlayer}{background}
  \foreach \x/\y/\z in {2/0/0,2/0/3,2/0/6,4/0/0,4/0/3,4/0/6,
    5/0/3,5/0/6,5/0/9,7/0/3,7/0/6,7/0/9,
    8/0/6,8/0/9,8/0/12,10/0/6,10/0/9,10/0/12,
    11/0/9,11/0/12,11/0/15,13/0/9,13/0/12,13/0/15,
    14/0/12,14/0/15,14/0/18,16/0/12,16/0/15,16/0/18,
    17/0/15,17/0/18,17/0/21,19/0/15,19/0/18,19/0/21,
    20/0/18,20/0/21,20/0/24,22/0/18,22/0/21,22/0/24} {
    \draw (v\x\y) -- (c\z);
  }
        
\foreach \x in {2,4,5,7,8,10,11,13,17,19,20,22} {
  \draw[thick] (\x,0) -- (\x,-3);
  \node[diamond,draw=black,fill=red,inner sep=0pt,minimum size=5pt] at
  (\x,-1.5) {};
}
\end{pgfonlayer}
\draw [gray,thick,decorate,decoration={brace,amplitude=5pt}]
  (23,-6) -- (1,-6)
  node [black,midway,below=3pt] {$2L+1$};

\foreach \x in {1.5,4.5,7.5,10.5,19.5,22.5} {
  \draw[dashed,gray] (\x,3) -- (\x,-6);
}

\begin{scope}[yshift=10cm]
  \begin{scope}[yshift=-15cm] 
  \foreach \x in {3,6,9,12,15,18,21} {
    \draw[fill=black] (\x,10)+(-5pt,-5pt) rectangle +(5pt,5pt)
    node[left=4pt,below=7pt,outer sep=0pt,inner sep=0pt] (c\x) {};
    \node[outer sep=0pt,inner sep=0pt] (cb\x) at (\x,10) {};
  }
  \foreach \x in {2,4,5,7,8,10,11,13,14,16,17,19,20,22} {
    \foreach \y in {12} {
      \shade[ball color=blue] (\x,\y) circle (6pt) node[outer
      sep=0pt,inner sep=0pt] (vb\x\y) {};
    }	
  }
  \begin{pgfonlayer}{foreground}
    \foreach \x in {14.5,15,15.5} {
      \foreach \y in {10,12} {
        \draw[fill=black] (\x,\y) circle (0.5pt);
      }
    }
  \end{pgfonlayer}
  \draw[white,fill=white] (13.75,8) rectangle +(2.5,4.5);
  \begin{pgfonlayer}{background}
    \foreach \x/\y/\z in {2/12/3,4/12/3,
    5/12/6,7/12/6,
    8/12/9,10/12/9,
    11/12/12,13/12/12,
    17/12/18,19/12/18,
    20/12/21,22/12/21} {
      \draw (vb\x\y) to              (cb\z);
      \draw (vb\x\y) to [bend left ] (cb\z);
      \draw (vb\x\y) to [bend right] (cb\z);
    }
  \end{pgfonlayer}
\end{scope}

\foreach \x in {3,6,9,12,15,18,21} {
  \draw[fill=black] (\x,2)+(-5pt,-5pt) rectangle +(5pt,5pt)
  node[left=4pt,below=7pt,outer sep=0pt,inner sep=0pt] (c\x) {};
  \node[outer sep=0pt,inner sep=0pt] (c\x) at (\x,2) {};
}
\foreach \x in {2,4,5,7,8,10,11,13,14,16,17,19,20,22} {
  \foreach \y in {0} {
    \shade[ball color=blue] (\x,\y) circle (6pt) node[outer
    sep=0pt,inner sep=0pt] (v\x\y) {};
  }	
}
\begin{pgfonlayer}{foreground}
  \foreach \x in {14.5,15,15.5} {
    \foreach \y in {-1.5,0,2} {
      \draw[fill=black] (\x,\y) circle (0.5pt);
    }
  }
\end{pgfonlayer}
\draw[white,fill=white] (13.75,-0.5) rectangle +(2.5,4.5);
      
  \begin{pgfonlayer}{background}
    \foreach \x/\y/\z in {2/0/3,4/0/3,
    5/0/6,7/0/6,
    8/0/9,10/0/9,
    11/0/12,13/0/12,
    17/0/18,19/0/18,
    20/0/21,22/0/21} {
      \draw (v\x\y) to		    (c\z);
      \draw (v\x\y) to [bend left ] (c\z);
      \draw (v\x\y) to [bend right] (c\z);
    }
        
\foreach \x in {2,4,5,7,8,10,11,13,17,19,20,22} {
  \draw[thick] (\x,0) -- (\x,-3);
  \node[diamond,draw=black,fill=red,inner sep=0pt,minimum size=5pt] at
  (\x,-1.5) {};
}
\end{pgfonlayer}

\foreach \x in {1.5,4.5,7.5,10.5,19.5,22.5} {
  \draw[dashed,gray] (\x,3) -- (\x,-6);
}

\end{scope}
\end{tikzpicture}

  \caption{Protograph of the joint decoder. Shown above are $2L+1$
    copies of the protograph of the joint decoder for a $(3,6)$
    regular LDPC code. The bottom graph shows the protograph of the
    joint decoder for the corresponding spatially coupled code.}
  \label{fig:protograph}
\end{figure}
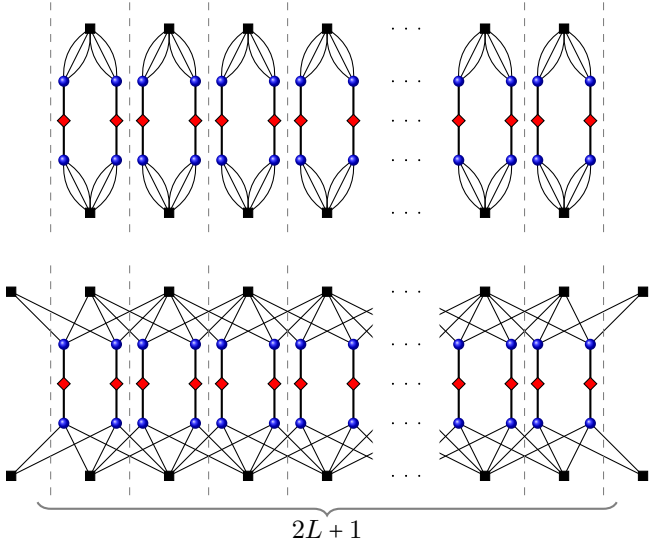

Although this ensemble is very instructive in understanding the
universality of \sc{} codes, the EBP curves for this ensemble exhibit
wiggles around the MAP threshold (similar to the single user channels
as discussed in \cite{Kudekar-it11}). The magnitude of these wiggles
appears to remain constant with increasing $L$ and their presence
implies that the BP threshold is smaller than the MAP threshold of the
underlying ensemble. Therefore, the $(3,6,L)$ ensemble does not
exhibit the threshold saturation phenomenon exactly. To overcome this,
we use the $(l,r,L,w)$ ensemble introduced in \cite{Kudekar-it11} for
the remainder of this work.


\subsection{The $(l,r,L,w)$ ensemble}
\label{sec:l-r-l-1}
The $(l,r,L,w)$ \sc{} ensemble can be described as
follows: Place $M$ variable nodes at each position in $[-L,L]$. The
check nodes are placed at positions $[-L,L+w-1]$, with $\frac lr M$
check nodes at each position. The connections are made as described in
\cite{Kudekar-it11}. This procedure generates a Tanner graph for the
$(l,r,L,w)$ ensemble. The design rate for the $(l,r,L,w)$ ensemble is
shown in \cite{Kudekar-it11} to be
\begin{align*}
  \rate(l,r,L,w) &= \left(1 - \frac lr\right) - \frac lr\frac{(w+1)
      - 2\sum_{i=0}^{w}
        \left(\frac{i}{w}\right)^{r}}{2L + 1}.
\end{align*}

Two such graphs (generated by the above procedure) are taken and the
variable nodes (of each graph) at each position are connected by a
random (uniform) permutation of size $M$ via channel nodes. This
procedure ensures that all the variable node positions are symmetric
and enables us to write down the density evolution (DE) equations in a
simple manner, as described in the following section.

\subsection{Density evolution of the $(l,r,L,w)$ ensemble}
\label{sec:density-evolution}

Let $\a_{i}^{(\ell)}$ and $\b_{i}^{(\ell)}$ denote the average density
emitted by the variable node at position $i$, at iteration $\ell$, for
codes $1$ and $2$ respectively. Set $\a_{i}^{(\ell)} = \b_{i}^{(\ell)}
= \Delta_{+\infty}$ for $i\notin[-L,L]$. The channel densities for
codes $1$ and $2$ are denoted by $\a_{\text{BMSC}}$ and
$\b_{\text{BMSC}}$ respectively. All the above densities are
$L$-densities conditioned on the transmission of the all-zero codeword
(see Section~\ref{sec:dens-evol-gexit}). We consider the parallel
schedule for each user (as described in \cite{Kudekar-it11}) and
update the correlation nodes before proceeding to the next
iteration. Let us define
\begin{align*}
  \mbox{\small $g(\x_{i-w+1},\cdots,\x_{i+w-1})\! \triangleq\! \left(\!\displaystyle\frac
    1w\!\sum_{j=0}^{w-1}\!\left(\!\frac 1w\!\displaystyle\sum_{k=0}^{w-1}\x_{i+j-k}\right)^{\boxast(r-1)}\right)^{\oast(l-1)}\!\!,$}
\end{align*}
\begin{align*}
    \mbox{\small $\Gamma(\x_{i-w+1},\cdots,\x_{i+w-1})\! \triangleq\! \left(\!\displaystyle\frac
    1w\!\sum_{j=0}^{w-1}\!\left(\!\frac
      1w\!\displaystyle\sum_{k=0}^{w-1}\x_{i+j-k}\right)^{\boxast(r-1)}\right)^{\oast
    l}\!\!.$}
\end{align*}
The DE equations for the joint \sc{} system can be written as
\begin{align*}
  \a_{i}^{(\ell+1)} &= f_{2\to 1}\left(\Gamma(\b_{i-w+1}^{(\ell)},\cdots,\b_{i+w-1}^{(\ell)}),\a_{\text{BAWGNMA}}\right)\oast \\
    &\phantom{==[}g(\a_{i-w+1}^{(\ell)},\cdots,\a_{i+w-1}^{(\ell)}),\\
  \b_{i}^{(\ell+1)} &= f_{1\to 2}\left(\Gamma(\a_{i-w+1}^{(\ell)},\cdots,\a_{i+w-1}^{(\ell)}),\a_{\text{BAWGNMA}}\right)\oast \\
    &\phantom{==[}g(\b_{i-w+1}^{(\ell)},\cdots,\b_{i+w-1}^{(\ell)}),
\end{align*}
for $i\in[-L,L]$. For a further discussion of the DE equations for the
$(l,r,L,w)$ \sc{} ensembles on BMS channels, see
\cite{Kudekar-istc10}. Using the notation $\underline{\a} \triangleq
(\a_{-L},\cdots,\a_{L})$, the fixed points of DE are given by
$(\a_{\text{BAWGNMA}},\underline{\a},\underline{\b})$, which satisfy
\begin{align}
\label{eq:de-sym}
  \a_{i} &= f\left(\Gamma(\b_{i-w+1},\cdots,\b_{i+w-1}),\a_{\text{BAWGNMA}}\right)\oast\notag\\
    &\phantom{==[}g(\a_{i-w+1},\cdots,\a_{i+w-1}) \notag\\
  \b_{i} &= f\left(\Gamma(\a_{i-w+1},\cdots,\a_{i+w-1}),\a_{\text{BAWGNMA}}\right)\oast\notag\\
    &\phantom{==[}g(\b_{i-w+1},\cdots,\b_{i+w-1}).
\end{align}
One can use the procedure outlined in \cite{Measson-it09}, known as
fixed-entropy DE, to compute both stable and unstable fixed points that
satisfy (\ref{eq:de-sym}).

\subsection{GEXIT curves for the $(l,r,L,w)$ ensemble}
\label{sec:gexit-curves-l}

Define the GEXIT value of a fixed point
$(\a_{\text{BAWGNMA}},\underline{\a},\underline{\b})$ by
\begin{align*}
  \G(\a_{\text{BAWGNMA}},\underline{\a},\underline{\b}) &\triangleq \frac{1}{2L+1}\sum_{i=-L}^{L}\G(\a_{\text{BAWGNMA}},\a_i,\b_i).
\end{align*}

The BP-GEXIT curve $\gexit(\alpha)$, for a fixed $A$, is the set of
points
$(\alpha,\G(\a_{\text{BAWGNMA}(\alpha)},\underline{\a},\underline{\b}))$. The
resulting curves for the \sc{} $(3,6,16,2)$ and $(3,6,32,4)$ ensembles
are shown in Fig.~\ref{fig:spatial-ebp-gexit-lrlw} for symmetric
channel conditions. These curves are very similar to the single user
case and demonstrate the phenomenon of threshold saturation at the
joint decoder, for symmetric channel conditions. For channel
parameters not on the symmetric line, these plots imply threshold
saturation towards the MAP boundary.


\begin{figure}[t!]
  \centering
  \begin{tikzpicture}[scale=0.48,>=stealth]

\definecolor{mycolor1}{rgb}{0,1,1}
\definecolor{mycolor2}{rgb}{1,0,1}

\begin{axis}[
scale only axis,
width=5in,
height=5in,
xmin=0, xmax=2,
ymin=-1.4, ymax=0,
xlabel=\mbox{\LARGE $\alpha$},
ylabel=\mbox{\LARGE $\mathsf{g}^{\text{BP}}$},
xmajorgrids,
ymajorgrids]

\addplot [
color=blue,
line width=1.25pt,
solid
]
coordinates{
 (0,-5.42634e-15)
 (0.01,-0.0288481)
 (0.02,-0.0576617)
 (0.03,-0.0864062)
 (0.04,-0.115047)
 (0.05,-0.143552)
 (0.06,-0.171886)
 (0.07,-0.200017)
 (0.08,-0.227914)
 (0.09,-0.255545)
 (0.1,-0.282881)
 (0.11,-0.309893)
 (0.12,-0.336553)
 (0.13,-0.362835)
 (0.14,-0.388714)
 (0.15,-0.414166)
 (0.16,-0.439169)
 (0.17,-0.463702)
 (0.18,-0.487746)
 (0.19,-0.511283)
 (0.2,-0.534298)
 (0.21,-0.556775)
 (0.22,-0.578701)
 (0.23,-0.600064)
 (0.24,-0.620855)
 (0.25,-0.641065)
 (0.26,-0.660686)
 (0.27,-0.679713)
 (0.28,-0.698141)
 (0.29,-0.715968)
 (0.3,-0.73319)
 (0.31,-0.749808)
 (0.32,-0.765823)
 (0.33,-0.781234)
 (0.34,-0.796047)
 (0.35,-0.810264)
 (0.36,-0.823888)
 (0.37,-0.836929)
 (0.38,-0.849388)
 (0.39,-0.861277)
 (0.4,-0.872598)
 (0.41,-0.883366)
 (0.42,-0.893589)
 (0.43,-0.903272)
 (0.44,-0.912423)
 (0.45,-0.921064)
 (0.46,-0.929198)
 (0.47,-0.936837)
 (0.48,-0.943992)
 (0.49,-0.950676)
 (0.5,-0.956902)
 (0.51,-0.962675)
 (0.52,-0.96802)
 (0.53,-0.972938)
 (0.54,-0.977447)
 (0.55,-0.981552)
 (0.56,-0.985281)
 (0.57,-0.988637)
 (0.58,-0.991625)
 (0.59,-0.994266)
 (0.6,-0.996567)
 (0.61,-0.998535)
 (0.62,-1.00019)
 (0.63,-1.00154)
 (0.64,-1.00258)
 (0.65,-1.00335)
 (0.66,-1.00386)
 (0.67,-1.0041)
 (0.68,-1.00408)
 (0.69,-1.00382)
 (0.7,-1.00332)
 (0.71,-1.0026)
 (0.72,-1.00166)
 (0.73,-1.00051)
 (0.74,-0.999149)
 (0.75,-0.997601)
 (0.76,-0.995861)
 (0.77,-0.99394)
 (0.78,-0.991841)
 (0.79,-0.989572)
 (0.8,-0.98714)
 (0.81,-0.984549)
 (0.82,-0.981801)
 (0.83,-0.978908)
 (0.84,-0.975869)
 (0.85,-0.972684)
 (0.86,-0.969363)
 (0.87,-0.965917)
 (0.88,-0.962338)
 (0.89,-0.958635)
 (0.9,-0.954805)
 (0.91,-0.950853)
 (0.92,-0.946785)
 (0.93,-0.942603)
 (0.94,-0.938315)
 (0.95,-0.933922)
 (0.96,-0.929425)
 (0.97,-0.924805)
 (0.98,-0.920068)
 (0.99,-0.915247)
 (1,-0.910315)
 (1.01,-0.90528)
 (1.02,-0.900199)
 (1.03,-0.894993)
 (1.04,-0.889678)
 (1.05,-0.884259)
 (1.06,-0.878783)
 (1.07,-0.87322)
 (1.08,-0.86755)
 (1.09,-0.861828)
 (1.1,-0.85597)
 (1.11,-0.850072)
 (1.12,-0.844073)
 (1.13,-0.838029)
 (1.14,-0.831877)
 (1.15,-0.825659)
 (1.16,-0.81937)
 (1.17,-0.813015)
 (1.18,-0.806605)
 (1.19,-0.800103)
 (1.2,-0.79354)
 (1.21,-0.786886)
 (1.22,-0.780155)
 (1.23,-0.773417)
 (1.24,-0.766598)
 (1.25,-0.759702)
 (1.26,-0.752815)
 (1.27,-0.745878)
 (1.28,-0.738858)
 (1.29,-0.731805)
 (1.3,-0.724716)
 (1.31,-0.71758)
 (1.32,-0.710319)
 (1.33,-0.70309)
 (1.34,-0.695805)
 (1.35,-0.688533)
 (1.36,-0.681237)
 (1.37,-0.673863)
 (1.38,-0.666452)
 (1.39,-0.659047)
 (1.4,-0.65163)
 (1.41,-0.644136)
 (1.42,-0.636657)
 (1.43,-0.62914)
 (1.44,-0.621651)
 (1.45,-0.61409)
 (1.46,-0.606589)
 (1.47,-0.598988)
 (1.48,-0.591449)
 (1.49,-0.583859)
 (1.5,-0.576268)
 (1.51,-0.568711)
 (1.52,-0.561142)
 (1.53,-0.553569)
 (1.54,-0.545893)
 (1.55,-0.53833)
 (1.56,-0.530777)
 (1.57,-0.523132)
 (1.58,-0.515497)
 (1.59,-0.507818)
 (1.6,-0.500099)
 (1.61,-0.492378)
 (1.62,-0.484632)
 (1.63,-0.476879)
 (1.64,-0.469073)
 (1.65,-0.461094)
 (1.66,-0.452923)
 (1.67,-0.444385)
 (1.68,-0.435389)
 (1.69,-0.424852)
 (1.69,-2.96196e-06)
 (1.7,-2.96196e-06)
 (1.71,-2.97647e-06)
 (1.72,-2.99078e-06)
 (1.73,-3.00487e-06)
 (1.74,-3.01874e-06)
 (1.75,-3.03239e-06)
 (1.76,-3.04579e-06)
 (1.77,-3.05894e-06)
 (1.78,-3.07183e-06)
 (1.79,-3.08444e-06)
 (1.8,-3.09678e-06)
 (1.81,-3.10882e-06)
 (1.82,-3.12056e-06)
 (1.83,-3.13199e-06)
 (1.84,-3.1431e-06)
 (1.85,-3.15387e-06)
 (1.86,-3.16429e-06)
 (1.87,-3.17436e-06)
 (1.88,-3.18406e-06)
 (1.89,-3.19339e-06)
 (1.9,-3.20233e-06)
 (1.91,-3.21087e-06)
 (1.92,-3.21901e-06)
 (1.93,-3.22672e-06)
 (1.94,-3.23401e-06)
 (1.95,-3.24086e-06)
 (1.96,-3.24726e-06)
 (1.97,-3.2532e-06)
 (1.98,-3.25868e-06)
 (1.99,-3.26369e-06)
 (2,-3.26821e-06)};

\addplot [
color=black,
fill=gray,fill opacity=0.2,
solid
]
coordinates{
 (0,-5.42634e-15)
 (0.01,-0.0288481)
 (0.02,-0.0576617)
 (0.03,-0.0864062)
 (0.04,-0.115047)
 (0.05,-0.143552)
 (0.06,-0.171886)
 (0.07,-0.200017)
 (0.08,-0.227914)
 (0.09,-0.255545)
 (0.1,-0.282881)
 (0.11,-0.309893)
 (0.12,-0.336553)
 (0.13,-0.362835)
 (0.14,-0.388714)
 (0.15,-0.414166)
 (0.16,-0.439169)
 (0.17,-0.463702)
 (0.18,-0.487746)
 (0.19,-0.511283)
 (0.2,-0.534298)
 (0.21,-0.556775)
 (0.22,-0.578701)
 (0.23,-0.600064)
 (0.24,-0.620855)
 (0.25,-0.641065)
 (0.26,-0.660686)
 (0.27,-0.679713)
 (0.28,-0.698141)
 (0.29,-0.715968)
 (0.3,-0.73319)
 (0.31,-0.749808)
 (0.32,-0.765823)
 (0.33,-0.781234)
 (0.34,-0.796047)
 (0.35,-0.810264)
 (0.36,-0.823888)
 (0.37,-0.836929)
 (0.38,-0.849388)
 (0.39,-0.861277)
 (0.4,-0.872598)
 (0.41,-0.883366)
 (0.42,-0.893589)
 (0.43,-0.903272)
 (0.44,-0.912423)
 (0.45,-0.921064)
 (0.46,-0.929198)
 (0.47,-0.936837)
 (0.48,-0.943992)
 (0.49,-0.950676)
 (0.5,-0.956902)
 (0.51,-0.962675)
 (0.52,-0.96802)
 (0.53,-0.972938)
 (0.54,-0.977447)
 (0.55,-0.981552)
 (0.56,-0.985281)
 (0.57,-0.988637)
 (0.58,-0.991625)
 (0.59,-0.994266)
 (0.6,-0.996567)
 (0.61,-0.998535)
 (0.62,-1.00019)
 (0.63,-1.00154)
 (0.64,-1.00258)
 (0.65,-1.00335)
 (0.66,-1.00386)
 (0.67,-1.0041)
 (0.68,-1.00408)
 (0.69,-1.00382)
 (0.7,-1.00332)
 (0.71,-1.0026)
 (0.72,-1.00166)
 (0.73,-1.00051)
 (0.74,-0.999149)
 (0.75,-0.997601)
 (0.76,-0.995861)
 (0.77,-0.99394)
 (0.78,-0.991841)
 (0.79,-0.989572)
 (0.8,-0.98714)
 (0.81,-0.984549)
 (0.82,-0.981801)
 (0.83,-0.978908)
 (0.84,-0.975869)
 (0.85,-0.972684)
 (0.86,-0.969363)
 (0.87,-0.965917)
 (0.88,-0.962338)
 (0.89,-0.958635)
 (0.9,-0.954805)
 (0.91,-0.950853)
 (0.92,-0.946785)
 (0.93,-0.942603)
 (0.94,-0.938315)
 (0.95,-0.933922)
 (0.96,-0.929425)
 (0.97,-0.924805)
 (0.98,-0.920068)
 (0.99,-0.915247)
 (1,-0.910315)
 (1.01,-0.90528)
 (1.02,-0.900199)
 (1.03,-0.894993)
 (1.04,-0.889678)
 (1.05,-0.884259)
 (1.06,-0.878783)
 (1.07,-0.87322)
 (1.08,-0.86755)
 (1.09,-0.861828)
 (1.1,-0.85597)
 (1.11,-0.850072)
 (1.12,-0.844073)
 (1.13,-0.838029)
 (1.14,-0.831877)
 (1.15,-0.825659)
 (1.16,-0.81937)
 (1.17,-0.813015)
 (1.18,-0.806605)
 (1.19,-0.800103)
 (1.2,-0.79354)
 (1.21,-0.786886)
 (1.22,-0.780155)
 (1.23,-0.773417)
 (1.24,-0.766598)
 (1.25,-0.759702)
 (1.26,-0.752815)
 (1.2629,-0.7508)
 (1.2629,0)
}|- (axis cs:0,0) -- cycle;

\addplot [
color=red,
solid
]
coordinates{
 (0,-5.40807e-15)
 (0.05,-0.143066)
 (0.1,-0.281892)
 (0.15,-0.412643)
 (0.2,-0.532204)
 (0.25,-0.638362)
 (0.3,-0.729847)
 (0.35,-0.806259)
 (0.4,-0.867919)
 (0.45,-0.915719)
 (0.5,-0.950897)
 (0.55,-0.974893)
 (0.6,-0.989277)
 (0.65,-0.995401)
 (0.7,-0.994695)
 (0.75,-0.988229)
 (0.8,-0.976937)
 (0.85,-0.961494)
 (0.9,-0.942417)
 (0.95,-0.920032)
 (1,-0.894493)
 (1.05,-0.865845)
 (1.1,-0.833976)
 (1.15,-0.798553)
 (1.2,-0.758906)
 (1.21,-0.750289)
 (1.22,-0.741353)
 (1.23,-0.732032)
 (1.24,-0.722031)
 (1.25,-0.710729)
 (1.26,-0.694181)
 (1.264,-0.145944)
 (1.264,-2.25967e-06)
 (1.28,-2.25967e-06)
 (1.29,-2.27725e-06)
 (1.3,-2.29482e-06)
 (1.35,-2.38248e-06)
 (1.4,-2.46971e-06)
 (1.5,-2.64175e-06)

};

\addplot [
color=brown,
solid
]
coordinates{
 (0,-5.37717e-15)
 (0.05,-0.142252)
 (0.1,-0.280282)
 (0.15,-0.410274)
 (0.2,-0.529128)
 (0.25,-0.634646)
 (0.3,-0.725561)
 (0.35,-0.801477)
 (0.4,-0.862707)
 (0.45,-0.910136)
 (0.5,-0.94499)
 (0.55,-0.968696)
 (0.6,-0.982814)
 (0.65,-0.988678)
 (0.7,-0.987717)
 (0.75,-0.980994)
 (0.8,-0.969437)
 (0.85,-0.953726)
 (0.9,-0.934386)
 (0.95,-0.911762)
 (1,-0.886036)
 (1.05,-0.857307)
 (1.1,-0.825512)
 (1.15,-0.790261)
 (1.2,-0.75017)
 (1.21,-0.741181)
 (1.22,-0.73167)
 (1.23,-0.721467)
 (1.24,-0.710095)
 (1.25,-0.696519)
 (1.26,-0.675946)
 (1.26,-2.24208e-06)
 (1.27,-2.24208e-06)
 (1.28,-2.25967e-06)
 (1.29,-2.27725e-06)
 (1.3,-2.29482e-06)
 (1.35,-2.38248e-06)
 (1.4,-2.46971e-06)
 (1.5,-2.64175e-06)

};

\node[rotate=90] (map) at (axis cs:1.315,-0.16) {\LARGE{$\bar{\alpha} = 1.2629$}};
\node[blue,rotate=90] (bp) at (axis cs:1.745,-0.16) {\LARGE{$\alpha^{\text{BP}} = 1.69$}};
\node[black] (area) at (axis cs:0.7,-0.35) 
{\Huge{$\int\limits_{\bar{\alpha}}^0\mathsf{g}^{\text{BP}}(\alpha) = 1$}};
\node[red] (g1) at (axis cs:0.7,-1.2) {\LARGE{BP-GEXIT,
    LDPC$(3,6,16,2)$}};
\draw[red,->] (g1) -- (axis cs:0.8,-0.976937);
\node[brown] (g2) at (axis cs:1.4,-1.05) {\LARGE{BP-GEXIT,
    LDPC$(3,6,32,4)$}};
\draw[brown,->] (g2) -- (axis cs:1.2,-0.75017);
\end{axis}

\end{tikzpicture}

  \caption{BP-GEXIT curves of the $(3,6,L,w)$ \sc{} LDPC and $(3,6)$
    regular LDPC ensembles for transmission over a $2$-user
    binary-input Gaussian MAC, with $A=1$. Also shown is the upper
    bound on the MAP threshold for the $(3,6)$ regular LDPC ensemble,
    computed using the area theorem.}
  \label{fig:spatial-ebp-gexit-lrlw}
\end{figure}
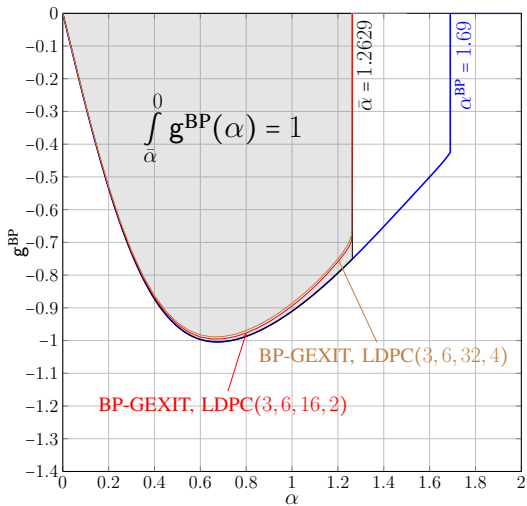


\section{Results and Concluding Remarks}
\label{sec:results-concl-remark}
It was shown in \cite{Kudekar-it11}, that for transmission over a BEC,
the BP threshold of \sc{} ensembles is essentially equal to the MAP
threshold of the underlying ensemble. This was observed numerically
for general BMS channels in \cite{Kudekar-istc10,Lentmaier-isit05}. It
was observed numerically for multi-user scenarios in
\cite{Yedla-isit11,Kudekar-isit11-MAC}. The notion of universality
with respect to channel parameters is important for multi-terminal
problems and has been discussed in
\cite{Yedla-isit11,Yedla-aller09,Yedla-istc10}.

In this paper, we study the 2-user binary-input Gaussian MAC and
observe that spatial coupling boosts the BP-ACPR of the joint decoder
to the MAP-ACPR of the underlying ensemble. The BP-ACPR for the
scenarios considered in this paper are shown in
Fig.~\ref{fig:roc_ldpc_gmac} and these results confirm the preliminary
results reported in \cite{Yedla-isit11}.  This figure shows that \sc{}
ensembles are near universal for this problem.  Based on the
observation that regular LDPC codes with large left degrees behave
like random codes and the fact that random codes are universal under
MAP decoding, we also conjecture that increasing the left degree
(keeping the rate constant) will push the MAP boundary towards the
boundary of the MAC-ACPR. An analytic proof of threshold saturation
remains an open problem. Such a proof would essentially show that it
is possible to achieve universality for the $2$-user binary-input
Gaussian MAC under iterative decoding.

\begin{figure}
  \centering
  \input{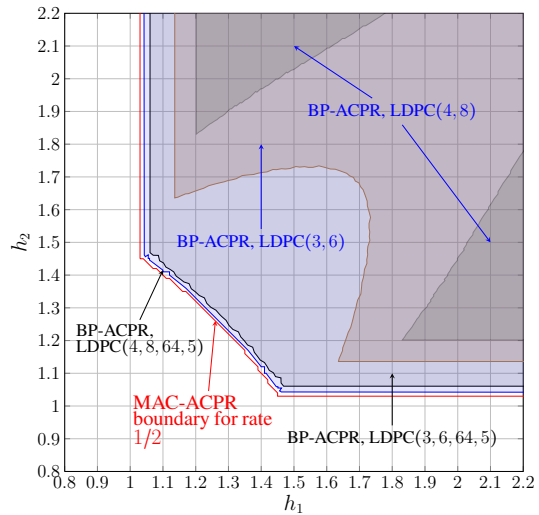}
  \caption{BP-ACPR of the $(3,6,64,5)$ and $(4,8,64,5)$ \sc{} LDPC
    ensembles for the $2$-user binary-input Gaussian MAC. Also shown
    are the BP-ACPRs for the $(3,6)$ and $(4,8)$ regular LDPC
    ensembles. The BP-ACPR of the $(4,8,64,5)$ \sc{} LDPC ensemble is
    very close to the MAC-ACPR, demonstrating the near-universal
    performance of \sc{} codes.}
  \label{fig:roc_ldpc_gmac}
\end{figure}

This work can be extended in a variety of ways.  For example, it is
straightforward to dispense with AWGN and compute the ACPRs of any
suitably parameterized 2-user binary-input MAC.  One can also
generalize these results to $m$-user MACs, larger input alphabets, and
multiple-input multiple-output (MIMO) systems.  In these cases, the
increase in computational complexity makes discretized DE infeasible
and Monte Carlo methods must be used to evaluate the DE and GEXIT
functions.  We conjecture that threshold saturation will continue to
occur for all these extensions and that \sc{} codes will achieve
near-universal performance.


\bibliographystyle{IEEEtran}
\bibliography{IEEEabrv,WCLabrv,WCLbib,WCLnewbib}

\end{document}